\documentclass[11pt, a4paper]{article}

\usepackage{amsmath,amsfonts,amstext,amscd,bezier,amssymb,amsthm,thmtools}
\usepackage{verbatim}
\usepackage{todonotes}
\usepackage{mathtools}
\usepackage[colorlinks,citecolor=blue,linkcolor=blue]{hyperref}
\usepackage{caption}
\usepackage{subcaption}

\newcommand{\ket}[1]{\big| #1 \big\rangle}
\newcommand{\bra}[1]{\big\langle #1 \big|}
\newcommand{\braket}[2]{\big\langle #1 \big| #2 \big\rangle}                 

\declaretheorem[style=definition,name=Definition, parent=section]{definition}

\declaretheorem[style=plain,name=Proposition, parent=section]{proposition}
\declaretheorem[style=plain,name=Corollary, parent=section]{corollary}

\newenvironment{demo}%
{\smallskip\noindent{\bf Proof: }}%
{\hfill$\square$}

\usepackage{soul}

\usepackage{authblk}
\begin{document}

\title{Walking on Vertices and Edges by \\ Continuous-Time Quantum Walk
}


\author[1]{Cau\^e F. Teixeira da Silva}
\author[2]{Daniel Posner}
\author[1,3]{Renato Portugal}

\affil[1]{\small National Laboratory of Scientific Computing (LNCC), 
Petr\'opolis, RJ, 25651-075, Brazil
}

\affil[2]{\small Universidade Federal Rural do Rio de Janeiro (UFRRJ), Seropédica, RJ, 23.897-000, Brazil}

\affil[3]{\small Senai Cimatec, Salvador, BA, 41650-010, Brazil}


\maketitle


\begin{abstract}

The quantum walk dynamics obey the laws of quantum mechanics with an extra locality constraint, which demands that the evolution operator is local in the sense that the walker must visit the neighboring locations before endeavoring to distant places. Usually, the Hamiltonian is obtained from either the adjacency or the Laplacian matrix of the graph and the walker hops from vertices to neighboring vertices. In this work, we define a version of the continuous-time quantum walk that allows the walker to hop from vertices to edges and vice versa.
As an application, we analyze the spatial search algorithm on the complete bipartite graph by modifying the new version of the Hamiltonian with an extra term that depends on the location of the marked vertex or marked edge, similar to what is done in the standard continuous-time quantum walk model.  We show that the optimal running time to find either a vertex or an edge is $O(\sqrt{N_e})$ with success probability $1-o(1)$, where $N_e$ is the number of edges of the complete bipartite graph.


\

\end{abstract}

\section{Introduction}

Quantum walks are the quantum counterpart of random walks~\cite{Aharanov1993}. Since random walks find many algorithmic applications, it is expected that quantum walks can also be used for finding interesting quantum algorithms~\cite{Portugal2018}. The quantum walk dynamics obey the laws of quantum mechanics with an extra locality constraint: if the position of the walker is at $x$, then the next position $x'$ must be at the neighborhood of $x$. The time evolution can be either discrete~\cite{AAKV01} or continuous~\cite{Farhi1998} but the spatial structure must be discrete. In the literature, we find many a kind of spatial structure: (1)~the vertex set in the continuous-time quantum walk~\cite{Farhi1998} and staggered quantum walk~\cite{Portugal2016}; (2)~the arc set in the coined quantum walk~\cite{HKSS14}; (3)~the edge set in Szegedy's model~\cite{Szegedy2004}; (4)~triangles in the quantum walk on simplicial complexes~\cite{Matsue2016}; (5)~edges connecting vertices and faces in the quantum walk on embeddings~\cite{Zhan2021}; and so on.


The spatial search algorithm by continuous-time quantum walks was introduced by Childs and Goldstone~\cite{CG04} after the successful definition of discrete-time spatial search on hypercubes~\cite{SKW03}. In the continuous-time case, the time evolution is driven by a Hamiltonian that is obtained from the adjacency matrix modified by a term that depends on the location of the marked vertex.
Experimental implementations of search algorithms by continuous-time quantum walk are described in~\cite{DGPSW20,WSXWJX20,BTPC21,QMWXWX22}.


The total graph $T(G)$ of a graph $G$ is defined by associating the vertices of $T(G)$ with both vertices and edges of $G$, so that two vertices of $T(G)$ are adjacent if and only if the corresponding vertices or edges of $G$ are adjacent or incident. The properties of the total graph of the complete bipartite graph have been extensively analyzed~\cite{Gao2012,Li2019,Zhao2022}. Besides, the total graph has been applied to communication and interconnection networks in order to achieve more robust fault-tolerant networks~\cite{Li2019,Dundar2004}. The spectral analysis of the adjacency matrix of total graphs of regular graphs has been studied in many papers~\cite{Cvetkovic1973,Liu2021,Hazama2002,Akbari2006}.



In this work, we address the problem of consistently defining a continuous-time quantum walk model that allows the walker to hop from vertices to edges and vice versa. An earlier attempt using discrete-time quantum walks revealed too difficult a task~\cite{Abreu_et_al_total_QW_2020}. There are many possible choices, including the case in which the walker is allowed to hop (1)~from a vertex to any neighboring vertex and to any incident edge, and (2)~from an edge $e$ to any neighboring edge and to any endpoint of $e$, which we call total quantum walk. Other cases are discussed in \autoref{SecConclusion}. The Hilbert space associated with a graph $G$ is spanned not only by the vertices but also by the edges and the dynamics is driven by a Hamiltonian defined as $H=\exp(-\textrm{i}\gamma A)$, where $A$ is the adjacency matrix of the total graph $T(G)$.
As an application, we analyze the spatial search algorithm on the complete bipartite graph by the total quantum walk. In this model, we are interested not only in determining the time complexity to find a marked vertex but also a marked edge. We show that the optimal running time to find either a vertex or an edge of the complete bipartite graph is $O(\sqrt{N_e})$ with success probability $1-o(1)$, where $N_e$ is the number of edges of $G$.


This work is organized as follows. In \autoref{SecReviewTotalGraph}, we review the concept of total graph $T(G)$ of a graph $G$.  In \autoref{SecEigenvaluesofTotalGraph}, we describe how to obtain the eigenvectors and eigenvalues of the total graph $T(G)$ using the eigenvectors and eigenvalues of a regular graph $G$. In \autoref{SecTotalQuantumWalkCT}, we define the total quantum walk. In \autoref{sec:method}, we outline a method to determine the computational complexity of continuous-time quantum-walk-based search algorithms on graphs with one marked vertex. In \autoref{SecQuantumSearchonKn}, we analyze quantum search by total quantum walk when the root graph is the complete graph $K_n$. In \autoref{SecQuantumSearchonKnn}, we analyze the quantum search by total quantum walk on the vertices and edges of the bipartite graph $K_{n,n}$. In \autoref{SecNumMet}, we present numerical results that confirms the hypotheses used to find the computational complexity. In \autoref{SecConclusion}, we present our final remarks.

\section{Total graph}\label{SecReviewTotalGraph}

There are several ways to generate new graphs from a root graph, such as, the line graph and the total graph~\cite{Cvetkovic09}. In the line graph case, the vertices represent the edges of its root graph so that the adjacency of the vertices of the line graph is obtained from the adjacency of the edges of the root graph. Formally, 
\begin{definition}
The \textit{line graph} $L(G)$ of a graph $G$ is a graph such that $V(L(G))=E(G)$ and two vertices of $L( G)$ are adjacent if and only if their corresponding edges share a common endpoint in $G$.
\end{definition}

The total graph generalizes the concept of line graph because the vertices of the total graph represent both edges and vertices of the original graph; the adjacency of the vertices of the total graph is obtained not only from the adjacency of the edges and vertices of the root graph but also from the incidence of edges on vertices. Formally,
\begin{definition}
The \textit{total graph} $T(G)$ of a graph $G$ is a graph such that $V(T(G))=V(G)\cup E(G)$ and two vertices of $T(G)$ are adjacent if and only if their corresponding elements are either adjacent in $G$ (when the corresponding elements are two vertices or two edges of $G$) or incident in $G$ (when the corresponding elements are a vertex and an edge of $G$).
\end{definition}

From the definition of total graph, if $V$ and $E$ are the subsets of $V(T(G))$ that represent vertices and edges of its root graph $G$, respectively, then the subgraph of $T(G)$ induced by $V$ is isomorphic to $G$ and the subgraph induced by $E$ is isomorphic to $L(G)$. Using the concept of total graph, we can define three related graph classes: $Q$-graphs, $R$-graphs, and subdivision graphs~\cite{CvetkoviC1975,Cvetkovic09}, which can be used to defined new quantum walk models whose walkers step on vertices and edges. A $Q$-graph $Q(G)$ is a subgraph of $T(G)$ obtained by removing the edges whose endpoints are vertices in $V$. A $R$-graph $R(G)$ is a subgraph of $T(G)$ obtained by removing the edges whose endpoints are vertices in $E$. A subdivision graph $S(G)$ is a subgraph of $T(G)$ obtained by removing the edges that have both endpoints in $V$ or have both endpoints in $E$.

The adjacency matrix of the total graph $T(G)$ is obtained from the adjacency matrices of the root graph $G$ and the line graph $L(G)$. Let us assume that the labels of the vertices of $T(G)$ are ordered so that the first $|V(G)|$ labels correspond to the vertices of $G$ and the remaining labels correspond to the edges of $G$. If we consider the submatrix of $A_{T(G)}$ with the labels only in $V$, we obtain the adjacency matrix $A_G$, and only in $E$, we obtain the adjacency matrix $A_{L(G)}$. If we partition $A_{T(G)}$ into blocks, it remains to consider the blocks with a label in $V$ and the other label in $E$. By the definition of total graph, the entries of these blocks are 1 when the corresponding edge is incident to the corresponding vertex, 0 otherwise. Note that this property is obeyed by the incidence matrix, whose definition is as follows.
\begin{definition}
The \textit{incidence matrix} of a graph $G$ is a $|V(G)|\times |E(G)|$ matrix $R_G$ such that 
\begin{equation*}
{\left(R_G\right)}_{ij}=\begin{cases}
1,& \textrm{if } e_j \text{ is incident to } v_i,\\
0,& \textrm{otherwise.}
\end{cases}
\end{equation*}
\end{definition}

\noindent
Now, using the label ordering presented above, we can obtain the adjacency matrix of the total graph using $A_G$, $A_{L(G)}$, and $R_G$, in the following way:
\begin{equation}\label{eqadjmattotgraph} 
A_{T(G)}=\begin{bmatrix}
A_G & R_G\\
R_G^T & A_{L(G)}
\end{bmatrix}.
\end{equation}
On the other hand, the incidence matrix $R$ obeys 
\begin{eqnarray*}
RR^T&=&A_G+rI,\\
R^TR&=&A_{L(G)}+2I.
\end{eqnarray*}

\subsection{The total graph of a regular graph}\label{SecEigenvaluesofTotalGraph}

In this subsection, we focus on the total graph of regular graphs with the goal of obtaining the eigenvalues and eigenvectors of their adjacency matrices by using our knowledge about the root graph. In fact, we obtain the spectrum of the total graph by using the following theorem:
\begin{proposition}[Cvetkovi\'c~\cite{Cvetkovic1973}]
	Let $G$ be a $r$-regular graph with $m$ edges and $n$ vertices and let $\{\lambda_i:i=1,...,n\}$ be the set of eigenvalues of $A_G$. The eigenvalues of the adjacency matrix of $T(G)$ are $(-2)$ with multiplicity $m-n$ and 
	\begin{equation*}
	\theta_i^\pm \,=\, \lambda_i-1+\frac{r}{2}\pm\sqrt{\lambda_i+1+\frac{r^2}{4}}.
	\end{equation*}
for $i\in \{1,...,n\}$ with multiplicity 1.
\end{proposition}

The next step is to obtain the eigenvectors of the total graph. For regular root graphs, Liu and Wang~\cite{Liu2021} describes how to obtain the eigenvectors of the Laplacian matrix of the total graph from the eigenvectors of the root graph and an orthonormal basis of the kernel of the incidence matrix of the total graph. Using Liu and Wang's theorem, we obtain the following similar corollary:
\begin{corollary}\label{corollaryEigenTotalGraph}
	Let $G$ be a $r$-regular connected graph with $n$ vertices, $m$ edges, and $r\geq 2$. 
	Let $r=\lambda_1\geq\dots\geq\lambda_n$ and $\textbf{v}_1,\dots,\textbf{v}_n$ be the eigenvalues and orthonormal eigenvectors of $A(G)$, respectively. 
	Let $\{\textbf{y}_1,\dots,\textbf{y}_{m-n}\}$ be an orthonormal basis of the kernel of $R_G$.
	Define 
\begin{equation*}
\zeta_i^{\pm}= \lambda_i-1+\frac{r}{2}\pm\sqrt{\lambda_i+1+\frac{r^2}{4}}, 
\end{equation*}
and
\begin{equation*}
\textbf{X}_i^{\pm}=\frac{1}{\sqrt c}\left(\begin{matrix}
(2-r-\lambda_i+\zeta^\pm_i)\textbf{v}_i\\
R_G^T\textbf{v}_i
\end{matrix}\right), 
\end{equation*}
where $i\in\{1, \dots,n\}$ and
\[c=(2-r-\lambda_i+\zeta^\pm_i)^2+\lambda_i+r.\]

\begin{itemize}
\item[(a)] If $G$ is non-bipartite, the eigenvalues of $T(G)$ are $\zeta_i^\pm$, $i\in\{1,\dots,n\}$, and $(-2)$ with multiplicity $m-n$, and the corresponding eigenvectors are $\textbf{X}_i^\pm$ and 
\begin{equation*}
\textbf{Y}_j=\left(\begin{matrix}
\textbf{0}\\
\textbf{y}_j
\end{matrix}\right),
\end{equation*}
where $j\in\{1,\dots,m-n\}$.

\item[(b)] If $G$ is bipartite, where $V(G)=V_1\cup V_2$ are the bipartitions of the vertices set, the eigenvalues of $T(G)$ are $\zeta_i^\pm$, $i\in\{1,\dots,n-1\}$, $(-r)$, and $(-2)$ with multiplicity $m-n+1$, and the corresponding eigenvectors are $\textbf{X}_i^\pm$, 
\begin{equation*}
\textbf{Z}=\left(\begin{matrix}
\textbf{J}_1\\
-\textbf{J}_2\\
0
\end{matrix}\right),
\end{equation*}
where $\textbf{J}_i$ is the $\left|V_i\right|$-dimensional vector with with all entries equal to 1,
and 
\begin{equation*}
\textbf{Y}_j=\left(\begin{matrix}
\textbf{0}\\
\textbf{y}_j
\end{matrix}\right), 
\end{equation*}
where $j\in\{1,\dots,m-n\}$.
\end{itemize}
\end{corollary}
\begin{demo}
Replace $\mu_i=r-\lambda_i$ and $\zeta_i^\pm=2r-\theta_i^\mp$ in Liu and Wang's theorem~\cite{Liu2021}.
\end{demo}

\section{Continuous-time quantum walk on vertices and edges}\label{SecTotalQuantumWalkCT}

The spatial structure of discrete-time quantum walks on graphs is (1)~the arc set in the coined model~\cite{HKSS14}, (2)~the edge set in Szegedy's model~\cite{Szegedy2004}, and the vertex set in the staggered model~\cite{Portugal2016}. On the other hand, the spatial structure of the continuous-time quantum walk is the vertex set~\cite{Farhi1998}. In this section, we define a continuous-time quantum walk that takes place on both vertices and edges in a mathematically consistent way. We name it \textit{total} quantum walk.

Let $G$ be a graph with vertex set $V(G)$ and edge set $E(G)$. We associate with this graph a $(|V(G)|+|E(G)|)$-dimensional Hilbert space spanned by $\big\{\ket{v}: v\in V(E)\cup E(G)\big\}$.
The evolution operator of the total quantum walk is obtained from the adjacency matrix $A_{T(G)}$ of the total graph $T(G)$, as follows
\begin{equation*}
U(t)=\mathrm{e}^{-\mathrm{i}\gamma H t}, 
\end{equation*}
where $\gamma$ is a real positive parameter and $H=A_{T(G)}$. In this model, if the walker is on a vertex $v$, after an infinitesimal time, the walker will stay put or hop to any neighboring vertex $v'$ or to any edge incident to $v$. If the walker is on an edge $e$, after an infinitesimal time, the walker will stay put or hop to any edge that has a common vertex with $e$ or to any vertex that is an endpoint of $e$.

The total quantum walk on a graph $G$ is equivalent to the standard continuous-time quantum walk on $T(G)$. The advantage of the new formulation consists in splitting the spatial part into two different types that can be exploited in applications, such as searching and decoherence analysis, in order to check whether vertices and edges behave in the same way or not.

 The spatial search algorithm is defined by modifying the model's Hamiltonian by adding an extra term that depends on the location of the marked vertex or edge $w$ as
\[H=-\gamma A-\ket{w}\bra{w}.\]
It is interesting to determine whether the time complexity of finding the marked element depends on its type.
The probability of finding a marked vertex at time $t$ is
\begin{equation}\label{eq:p(t)}
	p(t) = \left|\braket{w}{\psi(t)}\right|^2,
\end{equation}
where $\ket{\psi(t)}=U(t)\ket{\psi(0)}$ is the state of the quantum walk at time $t$ and $\ket{\psi(0)}$ is the initial state.
The goal of the algorithm is to find the optimal values of parameters $t$ and $\gamma$ so that the success probability is as high as possible.

\subsection{Method to find the computational complexity of search algorithms}\label{sec:method}

Ref.~\cite{Lugao2022,Bezerra21} has outlined a method for finding the computational complexity of spatial search algorithms when the graph has multiple marked vertices. In this subsection, we review this method for the single marked case, which is similar to the method described in~\cite{CGTX22}. For an operator $U$, let $\sigma(U)$ be the spectrum of $U$.
Suppose that the adjacency matrix $A$ of the total graph has exactly $q+1$ distinct eigenvalues $\phi_0>\phi_1>\dots>\phi_q$.
Let $P_{\ell}$ denote the orthogonal projector onto the eigenspace of $A$
for the eigenvalue $\phi_{\ell}$ for $0\leqslant \ell\leqslant q$, that is,
\begin{equation*}
	A = \sum_{\ell=0}^q \phi_{\ell} P_{\ell}.
\end{equation*}
Let $\lambda$ and $\ket{\lambda}$ be an eigenvalue and a normalized eigenvector of $H$, respectively. Hamiltonian $H$ and operator $-\gamma A$ may share common eigenvalues and eigenvectors, as shown in the following Proposition:

\begin{proposition}\label{prop:lambda_in_sigma_A}
$\lambda\in \sigma(-\gamma A)$ and $(-\gamma A)\ket{\lambda}=\lambda\ket{\lambda}$ if and only if $\braket{w}{\lambda}=0$.
\end{proposition}

The eigenvectors $\ket{\lambda}$ such that $\braket{w}{\lambda}=0$ play no role in the calculation of the probability $p(t)$, as can be seen from Eq.~\eqref{eq:p(t)}. The eigenvalues $\lambda$ associated with those eigenvectors are in the spectrum of $-\gamma A$.  Let us assume that $\braket{w}{\lambda}\neq 0$. Then, $\lambda\not\in \sigma(-\gamma A)$.

Using the definition of $H$ and assuming that $\lambda\not\in\sigma(-\gamma A)$, we obtain
\begin{equation}\label{masterequation}
	P_{\ell} \ket{\lambda}  = -\frac{\braket{w}{\lambda}}{\lambda+\gamma\phi_{\ell}}P_{\ell}\ket{w}.
\end{equation}
Using $\sum_{\ell=0}^q P_{\ell}=I$, we have 
\begin{equation*}
	\braket{w}{\lambda}=\sum_{\ell=0}^q \bra{w}P_{\ell}\ket{\lambda}.
\end{equation*}
Using Eq.~\eqref{masterequation}, we obtain
\begin{equation}\label{eq:sum_ell}
	\sum_{\ell=0}^q \frac{\left\|P_{\ell}\ket{w}\right\|^2}{\lambda+\gamma\phi_{\ell}}=-1.
\end{equation}

Usually, the spectral gap of the modified Hamiltonian tends to zero when the number of vertices increases. Let us assume that the asymptotic computational complexity of the search algorithm depends only on two eigenvalues $\lambda^{\pm}$ closest of $-\gamma\phi_0$ so that
\begin{equation}\label{eq:lambdapm}
    \lambda^{\pm} = -\gamma\phi_0 \pm {\epsilon} + O\big(\epsilon^2\big).
\end{equation}
In this case, the spectral gap is asymptotically $2\epsilon>0$. Eq.~(\ref{eq:sum_ell}) simplifies to
\begin{equation}\label{eq:lambda_exact}
    \sum_{\ell=0}^q \frac{\left\|P_{\ell}\ket{w}\right\|^2}{\gamma(\phi_{\ell}-\phi_0)\pm {\epsilon}}=-1.
\end{equation}
Up to second order in $\epsilon$ we obtain
\begin{equation}\label{eq:for_epsilon}
    {\bra{w}P_0\ket{w}} \pm \left(1  - \frac{S_1}{\gamma}\right){\epsilon} - \frac{ S_2}{\gamma^2}\,\epsilon^2=0,
\end{equation}
where
\begin{equation}\label{eq:S_1}
    S_1 = \sum_{\ell=1}^q \frac{\left\|P_{\ell}\ket{w}\right\|^2}{\phi_0-\phi_{\ell}}
\end{equation}
and 
\begin{equation}\label{eq:S_2}
    S_2 = \sum_{\ell=1}^q \frac{\left\|P_{\ell}\ket{w}\right\|^2}{(\phi_0-\phi_{\ell})^2}.
\end{equation}

Eq.~(\ref{eq:for_epsilon}) cannot have the linear term because that is the only way to have the $\pm \epsilon$ term in Eq.~\eqref{eq:lambdapm}. Then, 
\begin{equation}\label{eq:gamma}
    \gamma = S_1
\end{equation}
and
\begin{equation}\label{eq:epsilon}
    \epsilon= \frac{S_1\left\|P_0\ket{w}\right\|}{\sqrt{S_2}}.
\end{equation}

Using Eq.~\eqref{eq:p(t)} and an orthonormal set $\left\{\ket{\lambda}\right\}$ of eigenvectors of $H$, the probability of finding a marked vertex as a function of $t$ is 
\begin{equation*}
    p(t) =  \left| \sum_\lambda \mathrm{e}^{-\mathrm{i}\lambda t} \braket{w}{\lambda}\braket{\lambda}{\psi(0)}\right|^2 .
\end{equation*}
The exact eigenvalues $\lambda\not\in\sigma(-\gamma A)$ are roots of Eq.~(\ref{eq:lambda_exact}). In order to simplify the calculation, let us assume that the asymptotic success probability depends only on $\lambda^{\pm}$ and their associated eigenvectors $\ket{\lambda^{\pm}}$. Besides, let us assume that asymptotically
\begin{equation}\label{bra_lam_w}
    \braket{\lambda^+}{\psi(0)}\braket{w}{\lambda^+}=-\braket{\lambda^-}{\psi(0)}\braket{w}{\lambda^-}+o(1).
\end{equation}
Under those assumptions, the probability of finding a marked vertex reduces to
\begin{equation}\label{eq:p(t)=sin^2}
    p(t)=4\left|\braket{\lambda^+}{\psi(0)}\right|^2 \left|\braket{w}{\lambda^+}\right|^2\sin^2{\epsilon t} -o(1)+o(\epsilon t).
\end{equation}
Then, the optimal running time is
\begin{equation}\label{final_t_run}
    t_\text{opt} = \frac{\pi}{2\epsilon} 
\end{equation}
and the success probability is 
\begin{equation}\label{final_p_succ}
    p_{\mathrm{succ}} = 4\left|\braket{\lambda^+}{\psi(0)}\right|^2 \left|\braket{w}{\lambda^+}\right|^2 - o(1).
\end{equation}
There are many examples in the literature that obey the restrictions above~\cite{CG04,TSP22}.

Using Eq.~(\ref{masterequation}) and
\begin{equation*}
    \sum_{\ell=0}^q \left\| P_{\ell} \ket{\lambda^\pm } \right\|^2=1,
\end{equation*}
we obtain
\begin{equation*}
	\frac{1}{\left|\braket{w}{\lambda^\pm }\right|^2}  = \sum_{\ell=0}^q\frac{\left\|P_{\ell}\ket{w}  \right\|^2}{(\lambda^\pm +\gamma\phi_{\ell})^2}.
\end{equation*}
By re-scaling $\ket{\lambda^\pm }$ with a global phase, we obtain $\left|\braket{w}{\lambda^\pm }\right|=\braket{w}{\lambda^\pm }$. Taking the limit $\epsilon\xrightarrow{} 0$, we obtain
\begin{equation}\label{eq:wlambda}
    \braket{w}{\lambda^\pm } = \frac{S_1}{\sqrt{2S_2}}.
\end{equation}

The last missing quantity is $\braket{\lambda^\pm }{\psi(0)}$. If the total graph is regular and connected, the uniform superposition is an eigenvector of $\exp(-\mathrm{i}\gamma A)$ and $\exp(-\mathrm{i}\gamma\phi_0)$ is its eigenvalue with multiplicity one, where $\phi_0$ is the graph greatest eigenvalue. In this case, using Eq.~\eqref{masterequation} taking $\ell=0$ we obtain
\begin{equation*}
    \braket{\psi(0)}{\lambda^\pm } = -\frac{\braket{w}{\lambda^\pm }\braket{\psi(0)}{w}}{\lambda^\pm +\gamma\phi_{0}}.
\end{equation*}
Since all terms on the right-hand side are supposedly known, this means that we have a recipe to calculate $\braket{\psi(0)}{\lambda^\pm }$.  Assuming that $\ket{\psi(0)}$ is the normalized uniform state and by taking the limit $\epsilon\xrightarrow{} 0$, we obtain
\begin{equation}\label{eq:psi0lam}
     \braket{\psi(0)}{\lambda^\pm} = \mp\frac{1}{\sqrt{2N}  \left\| P_{0} \ket{w} \right\|},
\end{equation}
where $N$ is the number of vertices of the total graph. Note that Eq.~(\ref{bra_lam_w}) is consistent with Eqs.~(\ref{eq:wlambda}) and~(\ref{eq:psi0lam}).

If we use the amplitude amplification procedure~\cite{Portugal2018}, the time complexity of the search algorithm is $t_\text{opt}/\sqrt{p_\text{succ}}$. Using Eqs.~(\ref{final_t_run}) and~(\ref{final_p_succ}), the total running time  with success probability $\Omega(1)$ is
\begin{equation}
   t_\text{run}= \frac{\pi S_2}{2\,S_1^2} \sqrt{N}.
\end{equation}
Note that the time complexity of the search algorithm after using amplitude amplification does not depend on $\left\| P_{0} \ket{w} \right\|$. Besides, the time complexity is $O(\sqrt N)$ if $S_2$ has the same order of $S_1^2$.



\subsection{Quantum search by total quantum walk on the complete graph}\label{SecQuantumSearchonKn}

The class of complete graphs is the first candidate for the searching problem using the total quantum walk. Before starting to analyze the dynamics, let us consider the following result. Let $K_n$ be the complete graph with $n$ vertices. The total graph $T(K_n)$ is isomorphic to the Johnson graph $J(n+1,2)$ ($J(4,2)$ is illustrated in Fig.~\ref{fig:J(4,2)}) as shown in the next proposition:
\begin{proposition}
The total graph $T(K_n)$ of the complete graph $K_n$ is isomorphic to the Johnson graph $J(n+1,2)$.
\end{proposition}
\begin{proof}
Let us adopt the following convention: an arbitrary vertex of $T(K_n)$ has label $i$ or $\{i,j\}$ for $1\le i,j\le n$ depending on whether the corresponding element in the root graph $K_n$ is a vertex or an edge, respectively. An arbitrary vertex $v$ of $J(n+1,2)$ has label $v=\{i,j\}$ for $1\le i,j\le n+1$ and two vertices $v$ and $w$ are adjacent if and only if $|v\cap w|=1$. The isomorphism between $T(K_n)$ and $J(n+1,2)$ is established by mapping vertex $\{i,j\}$ of $T(K_n)$ into vertex  $\{i,j\}$ of $J(n+1,2)$ and by mapping vertex $i$ of $T(K_n)$ into vertex $\{i,n+1\}$ of $J(n+1,2)$. It is straightforward to check that two vertices are adjacent in $T(K_n)$ if and only if the corresponding vertices are adjacent in $J(n+1,2)$ by using the fact that $J(n+1,2)$ is the line graph of $K_{n+1}$.
\end{proof}

\begin{figure}[!ht]
\centering
\includegraphics[scale=0.5]{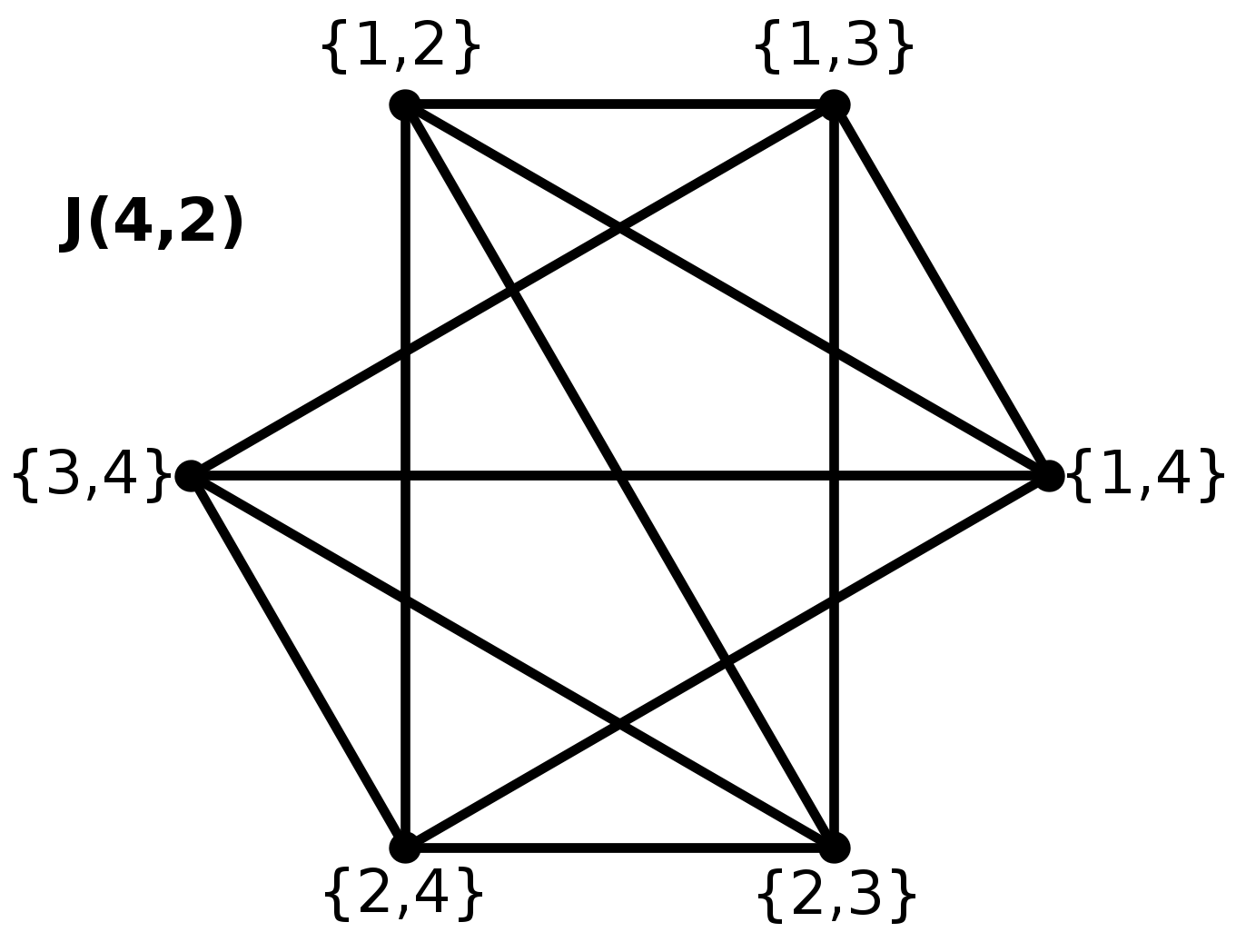}
\caption{The Johnson graph $J(4,2)$.}
\label{fig:J(4,2)}
\end{figure}


In the searching problem on the total quantum walk on the complete graph, the target element may be a vertex or an edge of $K_n$. The first conclusion we reach is that time complexity does not depend on the location of the marked element because  $T(K_n)$ is strongly regular~\cite{Liu2021}. Quantum search algorithm on $J(n+1,2)$ and strongly regular graphs were addressed in~\cite{TSP22,Janmark2014}, which has shown that 
the probability of finding the marked vertex of $J(n+1,2)\cong T(K_n)$ is
\begin{equation}
p(t)=\left(1-O\left(\dfrac{1}{n}\right)\right)\sin^2 \dfrac{t}{\sqrt{N }},
\end{equation}
where $N$ is the number of vertices plus the number of edges of $K_n$. Then, the optimal running time is 
\begin{equation}
t_\text{run}=\dfrac{\pi\sqrt{N}}{2},
\end{equation}
and the asymptotic success probability is 1.

\subsection{Quantum search by total quantum walk on the complete bipartite graph}\label{SecQuantumSearchonKnn}

In this subsection, we address the spatial search algorithm on the complete bipartite graph. 
The goal of the spatial search algorithm using the total quantum walk is to find either marked vertex or a marked edge.  

Let $K_{n,n}$ be the complete bipartite graph with $n$ vertices in each partition, so that $V(K_{n,n})=V_1\cup V_2$. The labels of the vertices of $V_1$ are $1,\dots,n$, and of $V_2$ are $1',\dots,n'$. The label of the edge connecting $i$ and $j'$ is $e_{i,j}$. 
The eigenvalues of $K_{n,n}$ are $n, 0, -n$ with multiplicities $1, 2n-2, 1$, respectively. A norm-1 eigenvector associated with eigenvalue $n$ is \begin{equation}\label{eigenvecKnn2}
\ket{v_n}=\displaystyle\dfrac{1}{\sqrt{2n}}\left(\sum_{i=1}^{n}\ket{i}+\ket{i'}\right).
\end{equation}
An orthonormal basis for the $0$-eigenspace is $\big\{\ket{v_k},\ket{v_{k'}}\big\}$ for $k\in\{1,...,n-1\}$, where
\begin{equation}\label{eigenvecKnn1}
\ket{v_k}=\displaystyle \sqrt{\dfrac{k}{k+1}}\left(\ket{k+1}-\sum_{i=1}^{k}\dfrac{1}{ k}\ket{i}\right)
\end{equation}
and
\begin{equation*}
\ket{v_{k'}}=\displaystyle \sqrt{\dfrac{k}{k+1}}\left(\ket{(k+1)'}-\sum_{i=1}^{k }\dfrac{1}{k}\ket{i'}\right).
\end{equation*}
A norm-1 eigenvector associated with eigenvalue $(-n)$ is 
\begin{equation*}
\ket{v_0}=\dfrac{1}{\sqrt{2n}}\left(\sum_{i=1}^{n}\ket{i}-\ket{i'}\right).
\end{equation*}
An orthonormal basis for the kernel of the incidence matrix $R_{K_{n,n}}$ is 
\begin{eqnarray}\label{ortbasisincKnn}
\ket{w_{i,j}}=\displaystyle\sqrt{h_{i,j}}\sum_{k_2=1}^{j+1}\sum_{k_1=1}^{i+1}\dfrac{(-i)^{\delta_{k_1,i+1}}(-j)^{\delta_{k_2,j+1}}}{ij}\ket{e_{k_1,k_2}},
\end{eqnarray}
where $h_{i,j}=\dfrac{ij}{ij+i+j+1}$ and $i,j\in\{1,...,(n-1)\}$.

Using this eigenbasis of $K_{n,n}$, \autoref{corollaryEigenTotalGraph}, and setting
\begin{equation*}
\theta_0^\pm\coloneqq\dfrac{n-2\pm\sqrt{n^2+4}}{2}
\end{equation*}
and
\begin{equation*}
\Delta_n^\pm\coloneqq\dfrac{n^2+4\pm(2-n)\sqrt{n^2+4}}{2},
\end{equation*}
the list of eigenvalues of $T(K_{n,n})$ in decreasing order is $\phi_0=2n$, $\phi_1=\theta_0^+$, $\phi_2=n-4$, $\phi_3=\theta_0^-$, $\phi_4=- 2$, and $\phi_5=-n$ with multiplicities $1$, $2(n-1)$, $1$, $2(n-1)$, $(n-1)^2$, and $1$, respectively. The list of eigenvectors using the notation of \autoref{corollaryEigenTotalGraph} is as follows.
A norm-1 eigenvector of eigenvalue $2n$ is 
\begin{equation}\label{eigenvecTKnnX+n}
\ket{X_n^+}=\dfrac{1}{\sqrt{4+2n}}\left(2\ket{v_n}+R_{K_{n,n}}^T\ket{v_n}\right).
\end{equation}
A norm-1 eigenvector of eigenvalue $(n-4)$ is 
\begin{equation*}
\ket{X^-_n}=\dfrac{1}{\sqrt{n^2+6n+4}}\left((-2-n)\ket{v_n}+R_{K_{n,n} }^T\ket{v_n}\right).
\end{equation*}
An orthonormal basis for the $\theta_0^\pm$-eigenspace is $\big\{\ket{X^\pm_{0,k}},\ket{X^\pm_{0,k'}}:k=1,\dots,n-1\big\}$, where
\begin{eqnarray*}
\ket{X^\pm_{0,k}}=\dfrac{1}{\sqrt{\Delta^\pm_n}}\left(\theta^\mp_n\ket{v_k}+R_{K_{n, n}}^T\ket{v_k}\right).
\end{eqnarray*}
A norm-1 eigenvector of eigenvalue $(-n)$ is
\begin{equation*}
\ket{Z}=\ket{v_0}.
\end{equation*}
An orthonormal basis for the $(-2)$-eigenspace is 
\begin{equation}\label{eigenvecTKnnYij}
\left\{\ket{Y_{i,j}}=\ket{w_{i,j}}:\,i,j=1,\dots,n-1\right\}.
\end{equation}

Let us turn our attention to finding the time complexity of the search algorithm. Since there are two possibilities for the marked element, we set $\ket{w^v}=\ket{n}$ for a marked vertex and $\ket{w^e}=\ket{e_{n,n}}$ for a marked edge, without loss of generality. We use the method described in the previous subsection. Since the total graph of $K_{n,n}$ is regular, we have 
\[
\left\|P_0\ket{w}\right\| =\frac{1}{\sqrt N},
\]
where $N=n^2+2n$ is the number of vertices of $T(K_{n,n})$.

By substituting~(\ref{eigenvecTKnnX+n}) to~(\ref{eigenvecTKnnYij}) into (\ref{eq:S_1}) and (\ref{eq:S_2}) with $q=5$, we obtain
\begin{equation}\label{sumSv1}
S^v_{1}=\dfrac{1}{2n}+\dfrac{5}{12{n}^{2}}+O\left(\dfrac{1}{n^3}\right),
\end{equation}
\begin{equation*}
S^e_{1}=\dfrac{1}{2n}+\dfrac{1}{2{n}^{2}}+O\left(\dfrac{1}{n^3}\right),
\end{equation*}
and
\begin{equation}\label{sumSe2}
S^v_{2}=S^e_{2}=\dfrac{1}{4{n}^{2}}+O\left(\dfrac{1}{n^3}\right).
\end{equation}

Using (\ref{eq:epsilon}), we obtain
\begin{equation}\label{eq:eps_v}
\epsilon^v=\left(1-\dfrac{7}{9n}+O\left(\dfrac{1}{n^2}\right)\right)\dfrac{1}{\sqrt{N}}
\end{equation}
and
\begin{equation}\label{eq:eps_e}
\epsilon^e=\left(1-\dfrac{1}{n}+O\left(\dfrac{1}{n^2}\right)\right)\dfrac{1}{\sqrt{N}}.
\end{equation}

Using (\ref{final_t_run}), the optimal running time in both cases is
\begin{equation}
t_\text{run}=\dfrac{\pi\sqrt{N}}{2}
\end{equation}
and using (\ref{final_p_succ}), the success probability of finding a vertex is
\begin{equation}
p^v_\text{succ}=1-\dfrac{14}{9n}+O\left(\dfrac{1}{n^2}\right)
\end{equation}
and of finding an edge is
\begin{equation}
p^e_\text{succ}=1-\dfrac{2}{n}+O\left(\dfrac{1}{n^2}\right).
\end{equation}
The difference between searching a vertex or an edge can only be seen in second order terms. 


\subsection{Numerical check}\label{SecNumMet}

Our method to find the computational complexity of the search algorithm on bipartite graphs using the total quantum walk employs two hypotheses that must be checked. In this section, we show numerically that we can rely on those hypotheses for total graph of complete bipartite graphs. Note that the computational effort to check the hypotheses is smaller than to check the final results (success probability and running time), because in the latter case we need to simulate the whole dynamics, while in the first case we basically calculate the eigenvectors and eigenvalues of the Hamiltonian. 

The first hypothesis is used in Eq.~(\ref{eq:lambdapm}), which assumes that the spectral gap tends to zero when the number of vertices increases. Besides, $\lambda^+$ and $\lambda^-$ are asymptotically symmetric about $-\gamma \phi_0$. This hypothesis is analyzed in the graphs of Fig.~\ref{Fig1}. The left-hand graph depicts $\pm\epsilon=\lambda^\pm+\gamma \phi_0$ as a function of $N$ when there is one marked vertex, and the right-hand graph when there is one marked edge. The numerical calculations show that $\epsilon$ tends to zero in a symmetric way. This indicates that it is reasonable to work under this hypothesis.

\begin{figure}[!ht]
\centering
\includegraphics[scale=0.5]{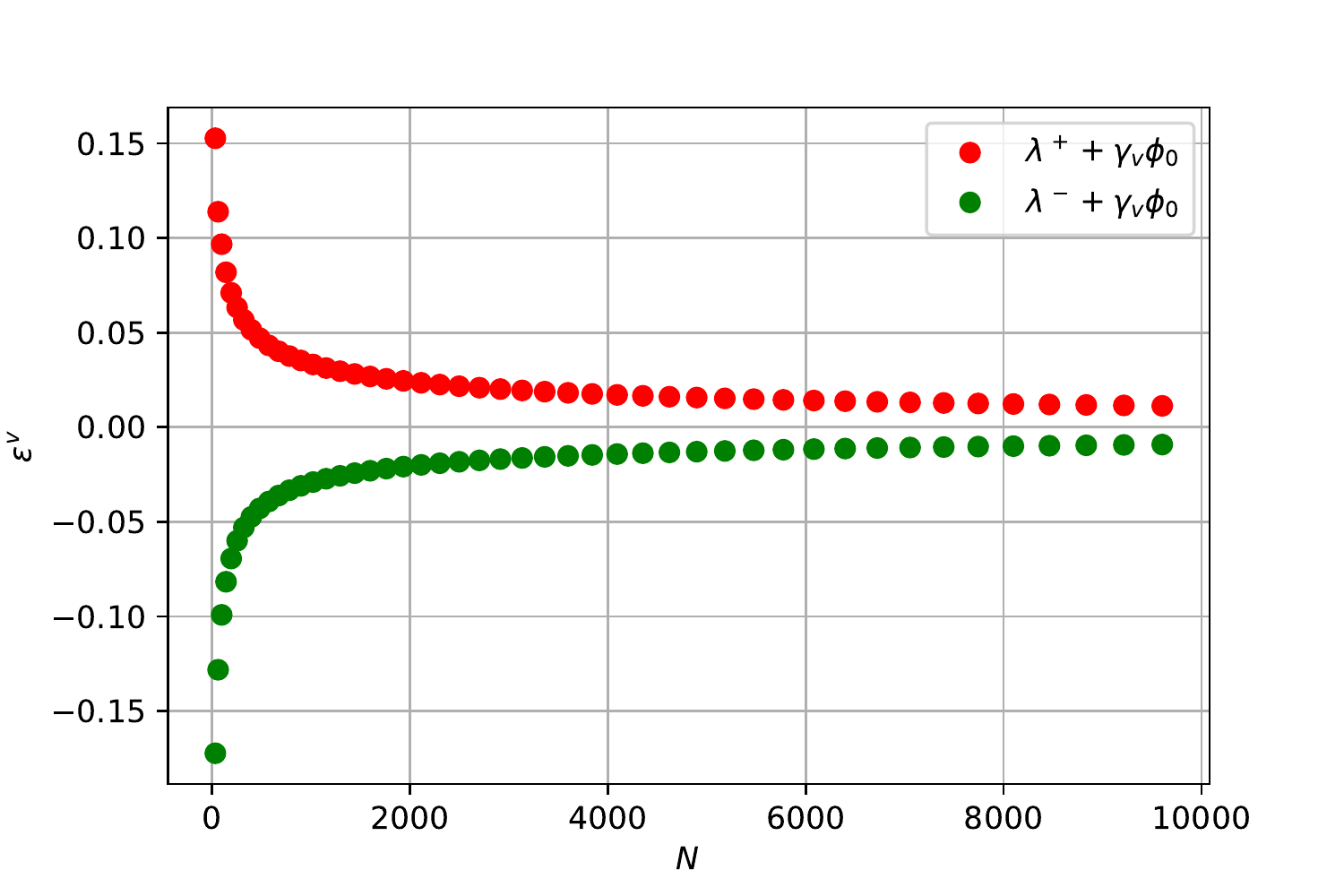}
\includegraphics[scale=0.5]{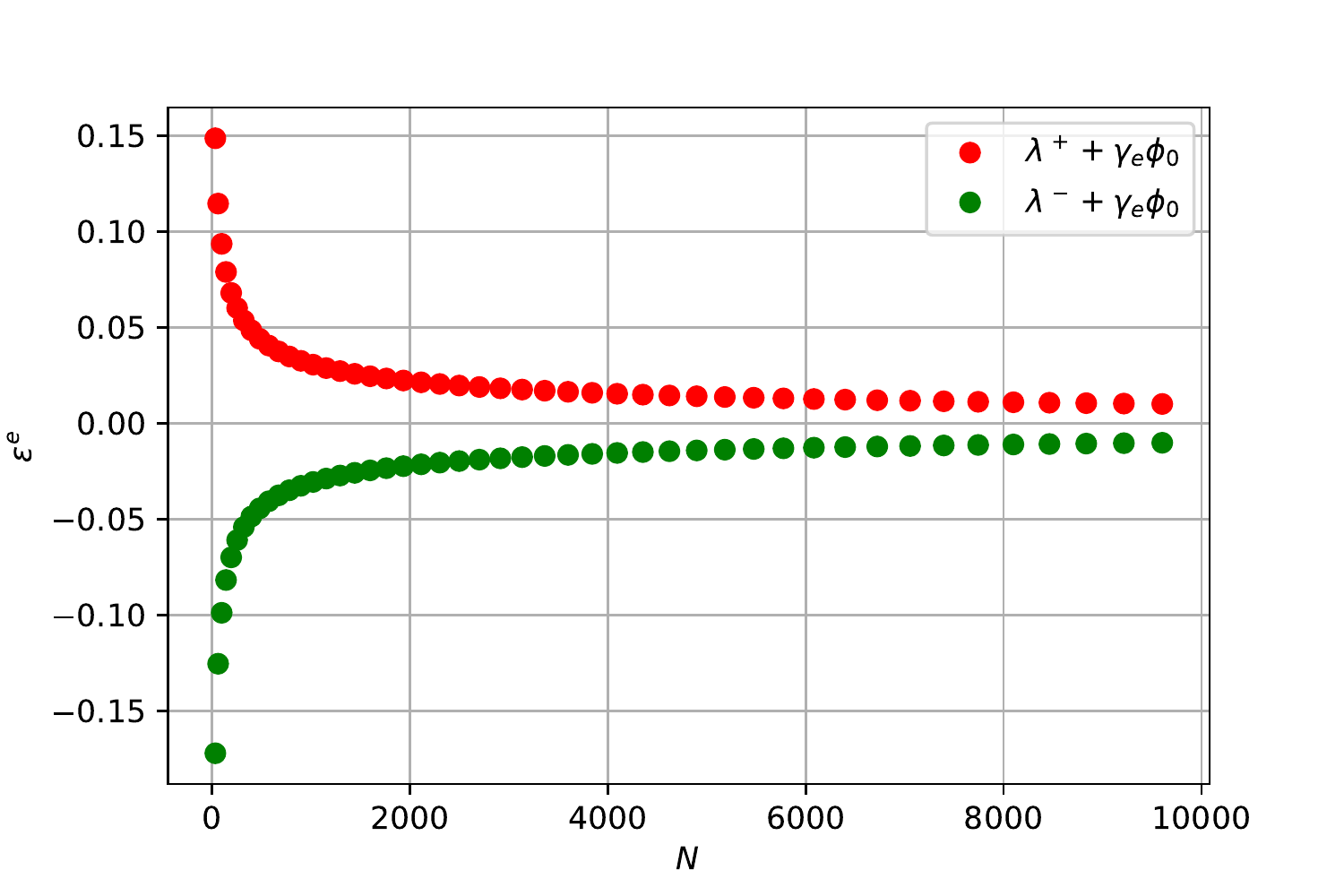}
\caption{Plot of $\pm\epsilon^v$ and $\pm\epsilon^e$ as a function of $N$ in order to check the first hypothesis. In the left-hand graph, we mark one vertex. In the right-hand graph, we mark one edge.}\label{Fig1}
\end{figure}

In order to help the checking of the first hypothesis, the left-hand graph of Fig.~\ref{Fig1a} depicts the loglog plot of the numerical calculation of $\epsilon$ (red dots) as a function of $N$ when there is one marked vertex together with the analytical results (black crosses) described by Eq.~(\ref{eq:eps_v}). The equation of the straight line that fits the numerical results is $\log_{10}\epsilon^v=-0.49\log_{10}N-0.079$, which is a good approximation to the behavior of $\epsilon^v$ as a function of $N$ given by Eq.~(\ref{eq:eps_v}).  The right-hand graph depicts the equivalent results when there is one marked edge. The numerical results matches the behavior of $\epsilon^e$ as a function of $N$ described by Eq.~(\ref{eq:eps_e}). Those graphs are good indications not only that the spectral gap tends to zero asymptotically but also that the algebraic results are correct. 

\begin{figure}[!ht]
\centering
\includegraphics[scale=0.5]{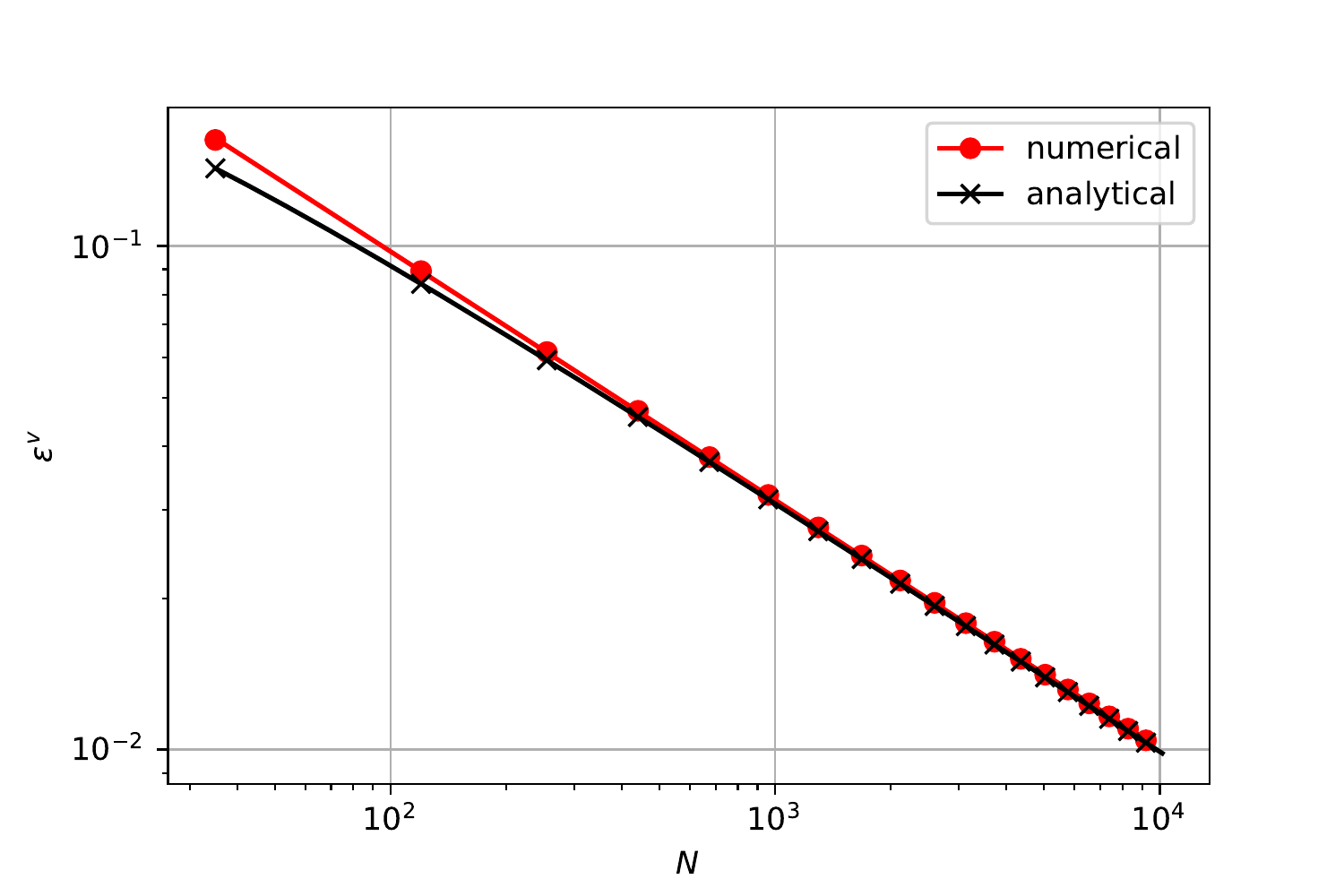}
\includegraphics[scale=0.5]{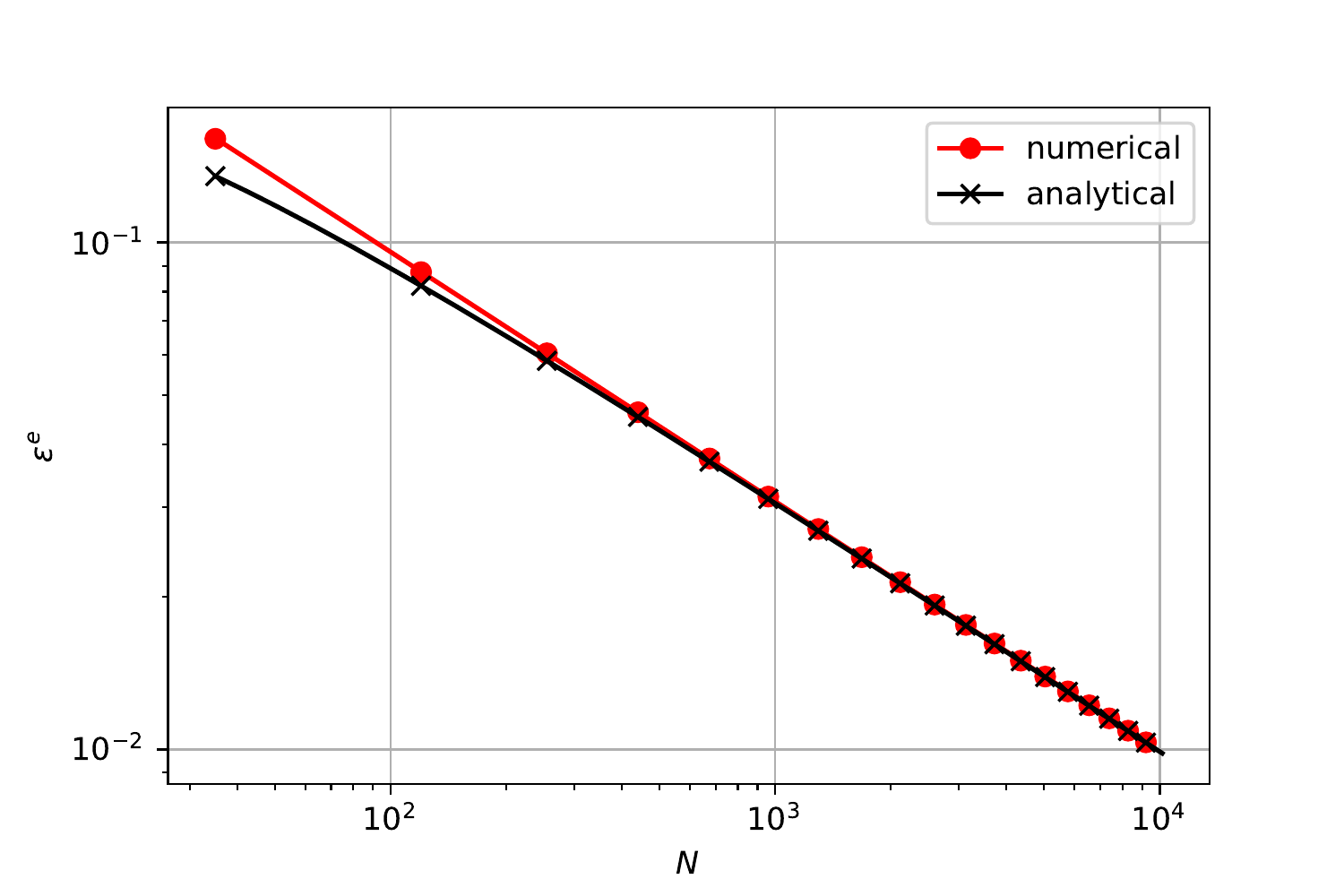}
\caption{Loglog plot of $\epsilon^v$ and $\epsilon^e$ as a function of $N$.}\label{Fig1a}
\end{figure}

The second hypothesis is used in Eq.~(\ref{bra_lam_w}), which assumes that the product of $\braket{\lambda^+}{\psi(0)}$ and $\braket{w}{\lambda^+}$, which are the overlaps of the initial state and the marked state with $\ket{\lambda^+}$, is equal to minus the product of $\braket{\lambda^-}{\psi(0)}$ and $\braket{w}{\lambda^-}$, asymptotically. This hypothesis is not strictly necessary in order to the method described in sec.~\ref{sec:method} to work. We assume it for simplicity, because the probability distribution as a function of $t$ becomes the square of a sinusoidal function. This hypothesis is analyzed in the graphs of Fig.~\ref{Fig2}. The left-hand graph depicts $\pm\braket{\lambda^\pm}{\psi(0)}\braket{w}{\lambda^\pm}$ as a function of $N$ when there is one marked vertex, and the right-hand graph when there is one marked edge. The numerical calculations again indicates that it is reasonable to work under this hypothesis.

\begin{figure}[!ht]
\centering
\includegraphics[scale=0.5]{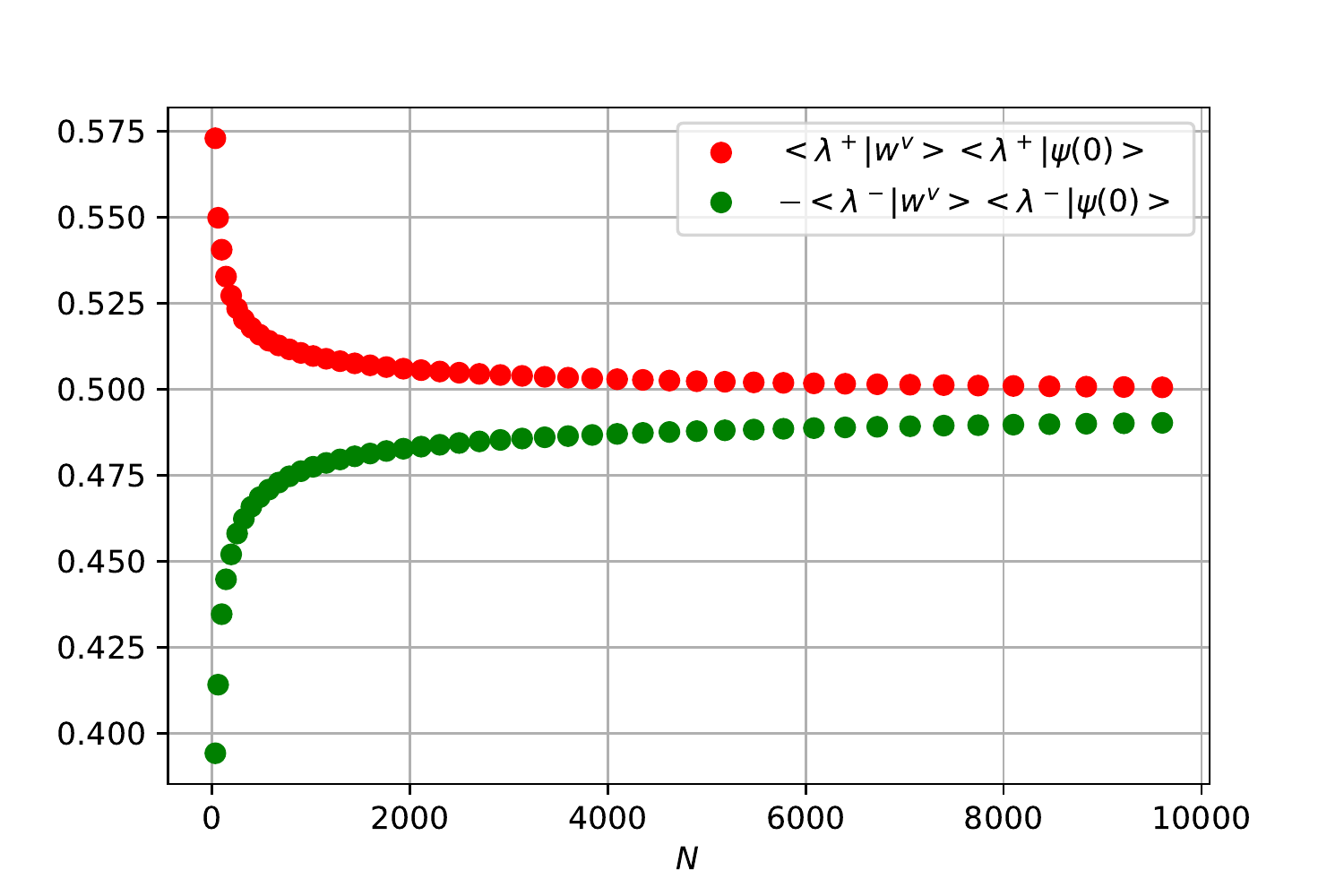}
\includegraphics[scale=0.5]{SymetryVofbraket+braket-nmin=5,nmax=100,salto=2.pdf}
\caption{Plot of the quantity $\pm\braket{\lambda^-}{\psi(0)}\braket{w}{\lambda^-}$ as a function of $N$ in order to check the second hypothesis. In the left-hand graph, we mark one vertex. In the right-hand graph, we mark one edge.}\label{Fig2}
\end{figure}

\section{Final remarks}\label{SecConclusion}

We have defined a continuous-time quantum walk version on a graph $G$, which allows the walker to hop from vertices to edges and vice versa. In our model, the walker may hop from (1)~vertex to vertex, (2)~vertex to edge, (3)~edge to vertex, and (4)~edge to edge.  To consistently define the evolution operator, we use the adjacency matrix of the total graph $T(G)$. Using this quantum walk model, we have shown that the spatial search algorithm on the complete bipartite graph can find a marked vertex or a marked edge in $O(\sqrt{N})$ steps with success probability $1-o(1)$.

There are at least three interesting extensions of the total quantum walk model on a graph $G$. In the first one, the walker may hop from (1)~vertex to edge, (2)~edge to vertex, and (3)~edge to edge. In this case, we have to use the adjacency matrix of the $Q$-graph $Q(G)$. In the second one, the walker may hop from (1)~vertex to vertex, (2)~vertex to edge, (3)~edge to vertex. In this case, we have to use the adjacency matrix of the $R$-graph $R(G)$.
 In the third one, the walker may hop from (1)~vertex to edge, and (2)~edge to vertex; it not allowed to hop from vertex to vertex or from edge to edge. In this case, we have to use the adjacency matrix of the subdivision graph $S(G)$. We intend to analyze those extensions in future works.


\section*{Data availability}
All data generated or analyzed during this study are included in this published article.

\section*{Acknowledgements}
The work of C. F. T. da Silva was supported by CNPq grant number 141985/2019-4. The work of R. Portugal was supported by FAPERJ grant number CNE E-26/202.872/2018, and CNPq grant number 308923/2019-7. The authors have no competing interests to declare that are relevant to the content of this article.


\end{document}